\documentclass{llncs}
\usepackage{epsfig}
\begin{document}

\title{Approximating minimum-power edge-multicovers}

\author{Nachshon Cohen \and Zeev Nutov}
\institute{The Open University of Israel,\\
\email{nachshonc@gmail.com, nutov@openu.ac.il}}

\maketitle

\begin{abstract}
Given a graph with edge costs, the {\em power} of a node is the
maximum cost of an edge incident to it, and the power of a graph is
the sum of the powers of its nodes. Motivated by applications in
wireless networks, we consider the following fundamental problem 
in wireless network design. Given a graph $G=(V,E)$ 
with edge costs and degree bounds $\{r(v):v \in V\}$, 
the {\sf Minimum-Power Edge-Multi-Cover} ({\sf MPEMC}) problem  is to find a
minimum-power subgraph $J$ of $G$ such that the degree of every node $v$ in $J$ 
is at least $r(v)$. We give two approximation algorithms for {\sf MPEMC},
with ratios $O(\log k)$ and $k+1/2$, where $k=\max_{v \in V} r(v)$ is the maximum 
degree bound.
This improves the previous ratios $O(\log n)$ and $k+1$, 
and implies ratios 
$O(\log k)$ for the {\sf Minimum-Power $k$-Outconnected Subgraph} and 
$O\left(\log k \log \frac{n}{n-k} \right)$ for the {\sf Minimum-Power $k$-Connected Subgraph}
problems; the latter is the currently best known ratio for the min-cost version
of the problem. 
\end{abstract}

\section{Introduction} \label{s:intro}

\subsection{Motivation and problems considered}

Wireless networks are studied extensively due to their wide applications. 
The power consumption of a station determines its transmission range, and thus also 
the stations it can send messages to; the power typically increases at least
quadratically in the transmission range. Assigning power levels to
the stations (nodes) determines the resulting communication network.
Conversely, given a communication network, the power required at $v$
only depends on the farthest node reached directly by $v$.
This is in contrast with wired networks, in which every pair of stations 
that communicate directly incurs a cost. 
An important network property is fault-tolerance, which is often measured
by minimum degree or node-connectivity of the network.
Node-connectivity is much more central here than edge-connectivity,
as it models stations failures. Such power minimization problems were vastly studied;
see for example \cite{ACMP,HKMN,LYN,N-dir,zeev-p} and the references therein for 
a small sample of papers in this area.
The first problem we consider is finding a low power network with 
specified lower degree bounds.
The second problem is the {\sf Min-Power $k$-Connected Subgraph} problem. 
We give approximation algorithms for these problems,
improving the previously best known ratios.

\begin{definition}
Let $(V,J)$ be a graph with edge-costs $\{c(e):e \in J\}$.
For a node $v \in V$ let $\delta_J(v)$ denote the set of edges incident to $v$ in $J$.
The {\em power} $p_J(v)$ of $v$ is the maximum cost of an edge in $J$ incident to $v$, 
or $0$ if $v$ is an isolated node of $J$;
i.e., $p_J(v) = \max_{e \in \delta_J(v)} c(e)$ if $\delta_J(v) \neq \emptyset$,
and $p_J(v)=0$ otherwise. For $V' \subseteq V$ the power of $V'$ w.r.t. $J$ 
is the sum $p_J(V')=\sum_{v \in V'} p_J(v)$ of the powers of the nodes in $V'$.
\end{definition}

Unless stated otherwise, all graphs are assumed to be undirected and simple.
Let $n=|V|$. 
Given a graph $G=(V,E)$ with edge-costs $\{c(e):e \in E\}$, 
we seek to find a low power subgraph $(V,J)$ of $G$ that satisfies some prescribed property.
One of the most fundamental problems in Combinatorial Optimization is finding 
a minimum-cost subgraph that obeys specified degree constraints 
(sometimes called also ``matching problems'') c.f. \cite{Sch}.
Another fundamental property is fault-tolerance (connectivity).
In fact, these problems are related, and we use our algorithm 
for the former as a tool for approximating the latter.

\begin{definition}
Given degree bounds $r=\{r(v):v \in V\}$, we say that an edge-set $J$
on $V$ is an {\em $r$-edge cover} if $d_J(v) \geq r(v)$ for every $v \in V$,
where $d_J(v)=|\delta_J(v)|$ is the degree of $v$ in the graph $(V,J)$.
\end{definition}


\noindent
{\sf Minimum-Power Edge-Multi-Cover} ({\sf MPEMC}): \\
{\em Instance:} 
A graph $G=(V,E)$ with edge-costs $\{c(e): e \in E\}$, 
degree bounds \hphantom{\em Instance: } 
$r=\{r(v):v \in V\}$. \\
{\em Objective:} 
Find a minimum power $r$-edge cover $J \subseteq E$.

\vspace{0.2cm}

Given an instance of {\sf MPEMC}, let $k=\max\limits_{v \in V} r(v)$ denote the maximum requirement.

We now define our connectivity problems. 
A  graph is {\em $k$-outconnected from $s$} if it contains $k$
internally-disjoint $sv$-paths for all $v \in V \setminus \{s\}$.
A graph is {\em $k$-connected} if it is $k$-outconnected from every node, 
namely, if it contains $k$ internally-disjoint $uv$-paths for all $u,v \in V$.

\vspace{0.2cm}

\noindent
{\sf Minimum-Power $k$-Outonnected Subgraph} ({\sf MP$k$OS}): \\
{\em Instance}: \ 
A graph $G=(V,E)$ with edge-costs $\{c(e): e \in E\}$, a root $s \in V$, and 
\hphantom{\em Instance: } an integer $k$. \\
{\em Objective}: 
Find a minimum-power $k$-outconnected from $s$ spanning subgraph $J$ 
\hphantom{\em Objective:} of $G$.

\vspace{0.2cm}

\noindent
{\sf Minimum-Power $k$-Connected Subgraph} ({\sf MP$k$CS}): \\
{\em Instance}: \ 
A graph $G=(V,E)$ with edge-costs $\{c(e): e \in E\}$ and an integer $k$. \\
{\em Objective}: 
Find a minimum-power $k$-connected spanning subgraph $J$ of $G$.

\subsection{Our Results}
The previous best approximation ratio for {\sf MPEMC} was $O(\log n)$ \cite{KMNT}. 
Our main result improves this ratio to $O(\log k)$.

\begin{theorem} \label{t:mpec}
{\sf MPEMC} admits an $O(\log k)$-approximation algorithm.
\end{theorem}

For small values of $k$, the problem admits also the ratios $k+1$
for arbitrary $k$ \cite{HKMN}, while for $k=1$ the best known ratio is $k+1/2=3/2$ \cite{KN-cov}. 
Our second result extends the latter ratio to arbitrary $k$.

\begin{theorem} \label{t:mpec'}
{\sf MPEMC} admits a $(k+1/2)$-approximation algorithm.
\end{theorem}

For small values of $k$, say $k \leq 6$, the ratio $(k+1/2)$ is better than
$O(\log k)$ because of the constant hidden in the $O( \cdot)$ term.
And overall, our paper gives the currently best known ratios for all values $k \geq 2$.

\vspace*{0.1cm}

In \cite{LYN} it is proved that an $\alpha$-approximation for {\sf MPEMC} implies an 
$(\alpha+4)$-approximation for {\sf MP$k$OS}.
The previous best ratio for {\sf MP$k$OS}  was 
$O(\log n)+4=O(\log n)$ \cite{LYN} for large values of $k=\Omega(\log n)$, 
and $k+1$ for small values of $k$ \cite{zeev-p}.
From Theorem~\ref{t:mpec} we obtain the following.

\begin{theorem} \label{c:mpoc}
{\sf MP$k$OS} admits an $O(\log k)$-approximation algorithm.
\end{theorem}

In \cite{HKMN} it is proved that an $\alpha$-approximation for {\sf MPEMC}
and a $\beta$-approximation for {\sf Min-Cost $k$-Connected Subgraph} implies a 
$(\alpha+2\beta)$-approximation for {\sf MP$k$CS}.
Thus the previous best ratio for {\sf MP$k$CS} was
$2\beta+O(\log n)$ \cite{KMNT}, where $\beta$ is the best ratio for {\sf MC$k$CS}
(for small values of $k$ better ratios for {\sf MP$k$CS} are given in \cite{zeev-p}).
The currently best known value of $\beta$ is $O\left(\log k \log \frac{n}{n-k} \right)$ \cite{zeevnew},
which is $O(\log k)$, unless $k=n-o(n)$.
From Theorem~\ref{t:mpec} we obtain the following.

\begin{theorem} \label{c:mpkc}
{\sf MP$k$CS} admits an $O(\beta+\log k)$-approximation algorithm,
where $\beta$ is the best ratio for {\sf MC$k$CS}. 
In particular, {\sf MP$k$CS} admits an $O\left(\log k \log \frac{n}{n-k} \right)$-approximation algorithm.
\end{theorem}

\subsection{Overview of the techniques}

Let the {\em trivial solution} for {\sf MPEMC} be obtained by picking for every 
node $v \in V$ the cheapest $r(v)$ edges incident to $v$.
It is known and easy to see that this produces an edge set of power at most 
$(k+1) \cdot {\sf opt}$, see \cite{HKMN}. 

Our $O(\log k)$-approximation algorithm uses the following idea.
Extending and generalizing an idea from \cite{KMNT},
we show how to find an edge set $I \subseteq E$ of power $O({\sf opt})$
such that for the residual instance, 
the trivial solution value is reduced by a constant fraction.
We repeatedly find and add such an edge set $I$ to the constructed solution, 
while updating the degree bounds accordingly to $r(v) \gets \max\{r(v)-d_I(v),0\}$.
After $O(\log k)$ steps, the trivial solution value is reduced to ${\sf opt}$, 
and the total power of the edges we picked is $O(\log k) \cdot {\sf opt}$.
At this point we add to the constructed solution the trivial solution of the residual problem,
which at this point has value ${\sf opt}$, obtaining an $O(\log k)$-approximate solution.

Our $(k+1/2)$-approximation algorithm uses a two-stage reduction.
The first reduction reduces {\sf MPEMC} to a constrained version of {\sf MPEMC} 
with $k=1$, where we also have lower bounds $\ell_v$ on the power 
of each node $v \in V$; these lower bounds are determined by the trivial solution
to the problem. We will show that a $\rho$-approximation 
algorithm to this constrained version implies a $(k-1+\rho)$-approximation
algorithm for {\sf MPEMC}. The second reduction reduces the constrained version 
to the {\sf Minimum-Cost Edge Cover} problem with a loss of $3/2$ in the 
approximation ratio. As {\sf Minimum-Cost Edge Cover} admits a polynomial time algorithm, 
we get a ratio $\rho=3/2$ for the constrained problem,
which in turn gives the ratio $k-1+\rho=k+1/2$ for {\sf MPEMC}.

\section{An $O(\log k)$-approximation (proof of Theorem~\ref{t:mpec})}

As in \cite{KMNT}, we reduce {\sf MPEMC} to {\sf Bipartite MPEMC},
where $G=(V,E)$ is a bipartite graph with sides $A,B$, and $r(a)=0$ for every $a \in A$ 
(so, only the nodes in $B$ may have positive degree bound).
This is done by taking two copies $A=\{a_v:v \in V\}$ and $B=\{b_v:v \in V\}$ of $V$,
for every edge $e=uv \in E$ adding the two edges $a_ub_v$ and $a_vb_u$ of cost $c(e)$ each,
and for every $v \in V$ setting $r(b_v)=r(v)$ and $r(a_v)=0$.
It is proved in \cite{KMNT} that this reduction invokes a factor of $2$ in the
approximation ratio, namely, that a $\rho$-approximation for bipartite {\sf MPEMC} 
implies a $2\rho$-approximation for general {\sf MPEMC}.

Let ${\sf opt}$ denote the optimal solution value of a problem instance at hand.
For $v \in V$, let $w_v$ be the cost of the $r(v)$-th least cost edge incident to $v$ 
in $E$ if $r(v) \geq 1$, and $w_v=0$ otherwise.
Given a partial solution $J$ to {\sf Bipartite MPEMC} let 
$r_J(v)=\max\{r(v)-d_J(v),0\}$ be the {\em residual bound} of $v$ w.r.t. $J$.
Let 
$$
R_J = \sum_{b \in B}{w_b r_J(b)} \ .
$$

The main step in our algorithm is given in the following lemma, 
which will be proved later.
 
\begin{lemma} \label{l:1}
There exists a polynomial time algorithm that given an edge set $J \subseteq E$, 
an integer $\tau$, and a parameter $\gamma>1$, either correctly establishes that 
$\tau < {\sf opt}$, or  
returns an edge set $I \subseteq E \setminus J$ such that 
$p_I(V) \leq (1+\gamma) \tau$ and $R_{J \cup I} \leq \theta R_J$, where
$\theta=1-\left(1-\frac{1}{\gamma}\right)\left(1-\frac{1}{e}\right)$. 
\end{lemma}

\begin{lemma} \label{l:2}
Let $J \subseteq E$ and let $F \subseteq E \setminus J$ be an edge set obtained by picking 
$r_J(b)$ least cost edges in $\delta_{E \setminus J}(b)$ for every $b \in B$. 
Then $J \cup F$ is an $r$-edge-cover and: $p_F(B) \leq {\sf opt}$,
$p_F(A) \leq R_J \leq k \cdot {\sf opt}$. 
\end{lemma}
\begin{proof}
Since $F$ is an $r_J$-edge-cover, $J \cup F$ is an $r$-edge-cover.
By the definition of $F$, for any $r$-edge-cover $I$,
$p_F(b) \leq w_b \leq p_I(b)$ for all $b \in B$.
In particular, if $I$ is an optimal $r$-edge-cover, then 
$$p_F(B) \leq \sum_{b \in B} w_b \leq \sum_{b \in B} p_I(b) = p_I(B) \leq {\sf opt} \ .$$
Also, 
$$R_J = \sum_{b \in B}{w_b r_J(b)} \leq k \cdot \sum_{b \in B} w_b \leq k \cdot {\sf opt} \ .$$
Finally, $p_F(A) \leq R_J$ since
$$p_F(A) = \sum_{a \in A} p_F(a) \leq \sum_{a \in A} \sum_{e \in \delta_F(a)} c(e)
=\sum_{e \in F} c(e) \leq  \sum_{b \in B} w_b r_J(b) =R_J \ . $$
This concludes the proof of the lemma.
\qed
\end{proof}

Theorem \ref{t:mpec} is deduced from Lemmas \ref{l:1} and \ref{l:2} as follows.
We set $\gamma$ to be constant strictly greater than $1$, say $\gamma=2$.
Then $\theta=1-\frac{1}{2}\left(1-\frac{1}{e}\right)$.
Using binary search, we find the least integer $\tau$ such that the following procedure 
computes an edge set $J$ satisfying $R_J \leq \tau$.  

\vspace{0.2cm}

\noindent
{\em Initialization:} $J \gets \emptyset$. \\
{\em Loop:} \ Repeat $\lceil \log_{1/\theta} k \rceil$ times: \\
\hphantom{Loop:} \ \ \ \ Apply the algorithm from Lemma~\ref{l:2}: \\ 
\hphantom{Loop:} \ \ \ \ - If it establishes that $\tau < {\sf opt}$ then return ``ERROR'' and STOP. \\
\hphantom{Loop:} \ \ \ \ - Else do $J \gets J \cup I$. 

\vspace{0.2cm}

After computing $J$ as above, we compute an edge set $F \subseteq E \setminus J$ as in Lemma~\ref{l:2}.
The edge-set $J \cup F$ is a feasible solution, by Lemma~\ref{l:2}.
We claim that for any $\tau \geq {\sf opt}$ the above procedure returns an edge 
set $J$ satisfying $R_J \leq \tau$; thus binary search indeed applies.
To see this, note that $R_\emptyset \leq k \cdot {\sf opt}$ and thus
$$R_J \leq R_\emptyset \cdot \theta^{\lceil \log_{1/\theta} k \rceil} \leq 
k \cdot {\sf opt} \cdot 1/k={\sf opt} \leq \tau \ .$$
Consequently, the least integer $\tau$ for which the above procedure does not return ``ERROR''
satisfies $\tau \leq {\sf opt}$. Thus
$p_J(V) \leq \lceil \log_{1/\theta} k \rceil \cdot (1+\gamma) \cdot \tau = O(\log k) \cdot {\sf opt}$.
Also, by Lemma~\ref{l:2}, $p_F(V) \leq {\sf opt}+R_J \leq 2{\sf opt}$. 
Consequently, 
$$p_{J \cup F}(V) \leq p_J(V)+p_F(V) =O(\log k) \cdot {\sf opt}+2{\sf opt} =O(\log k) \cdot {\sf opt} \ .$$

In the rest of this section we prove Lemma~\ref{l:1}. 
It is sufficient to prove the statement in the lemma for the residual 
instance $((V,E \setminus J),r_J)$ with edge-costs restricted to $E \setminus J$; 
namely, we may assume that $J=\emptyset$.
Let $R=R_\emptyset=\sum_{b \in B} w_b r(b)$.

\begin{definition}
An edge $e \in E$ incident to a node $b \in B$ is {\em $\tau$-cheap} if 
$c(e) \leq \frac{\tau\gamma}{R} \cdot w_b r(b)$.
\end{definition}

\begin{lemma} \label{l:cheap}
Let $F$ be an $r$-edge-cover, let $\tau \geq p_F(B)$, and let 
$$I=\bigcup\limits_{b \in B} \{e \in \delta_E(b):c(e) \leq \frac{\tau\gamma}{R} \cdot w_b r(b)\}$$ 
be the set of $\tau$-cheap edges in $E$.
Then $R_{I \cap F} \leq R/\gamma$ and $p_I(B) \leq \gamma \tau$.
\end{lemma}
\begin{proof}
Let $D=\{b \in B: \delta_{F \setminus I}(b) \neq \emptyset\}$. 
Since for every $b \in D$ there is an edge $e \in F \setminus I$ incident to
$b$ with $c(e) > \frac{\tau\gamma}{R} \cdot w_b r(b)$,
we have $p_{F \setminus I}(b) \geq \frac{\tau \gamma}{R} \cdot w_b r(b)$ for every $b \in D$.
Thus
$$\tau \geq p_F(B) \geq p_{F \setminus I}(B) = \sum_{b \in D} p_{F \setminus I}(b) \geq 
\tau \cdot \frac{\gamma}{R} \sum_{b \in D} w_b r(b)  \ .$$
This implies $\sum_{b \in D} w_b r(b) \leq R/\gamma$. 
Note that for every $b \in B\setminus D$,
$\delta_F(b) \subseteq \delta_I(b)$ and hence $r_{I\cap F}(b)=r_F(b)=0$.
Thus we obtain:
$$R_{I\cap F} = \sum_{b \in B}w_b r_{I\cap F}(b) = \sum_{b \in D} w_b r_{I \cap F}(b) \leq 
\sum_{b \in D} w_b r(b) \leq R/\gamma \ .$$ 

To see that $p_I(B) \leq \gamma \tau$ note that
$$p_I(B)=\sum_{b \in B}p_I(b) \leq \frac{\tau \gamma}{R} \sum_{b \in B} w_b r(b)=
\frac{\tau \gamma}{R} \cdot R = \tau \gamma \ .$$
This concludes the proof of the lemma.
\qed
\end{proof}

In \cite{KMNT} it is proved that the following problem,
which is a particular case of submodular function  minimization subject to 
matroid and knapsack constraint (see \cite{LNMN}) admits a 
$\left(1-\frac{1}{e}\right)$-approximation algorithm.

\vspace{0.2cm}

\noindent
{\sf Bipartite Power-Budgeted Maximum Edge-Multi-Coverage} ({\sf BPBMEM}): \\
{\em Instance:} \ A bipartite graph $G=(A \cup B,E)$ with edge-costs $\{c(e):e \in E\}$~and
\hphantom{\em Instance: } node-weights $\{w_v:v \in B\}$,  
degree bounds $\{r(v):v \in B\}$, and a \hphantom{\em Instance: } budget $\tau$.\\
{\em Objective:} Find $I \subseteq E$ with $p_I(A) \leq \tau$ that maxmizes
$${\sf val}(I)=\sum_{v \in B} w_v \cdot \min\{d_I(v),r(v)\} \ .$$

\vspace{0.2cm}

The following algorithm computes an edge set as in Lemma~\ref{l:1}. 
\begin{enumerate}
\item
Among the $\tau$-cheap edges, compute a $\left(1-\frac{1}{e}\right)$-approximate
solution $I$ to {\sf BPBMEM}.
\item
If $R_I \leq \theta R$ then return $I$, where 
$\theta =1-\left(1-\frac{1}{\gamma}\right)\left(1-\frac{1}{e}\right)$; \\
Else declare ``$\tau<{\sf opt}$''.
\end{enumerate}

\vspace{0.2cm}

Clearly, $p_I(A) \leq \tau$. By Lemma~\ref{l:cheap}, $p_I(B) \leq \gamma \tau$.
Thus $p_I(V) \leq p_I(A)+p_I(B) \leq (1+\gamma) \tau$.

Now we show that if $\tau \geq {\sf opt}$ then $R_I \leq \theta R$.
Let $F$ be the set of cheap edges in some optimal solution. 
Then $p_F(A) \leq {\sf opt} \leq \tau$. 
By Lemma~\ref{l:cheap} $R_F \leq R/\gamma$, namely, 
$F$ reduces $R$ by at least $R\left(1-\frac{1}{\gamma}\right)$.
Hence our $\left(1-\frac{1}{e} \right)$-approximate solution $I$ to {\sf BPBMEM}
reduces $R$ by at least $R\left(1-\frac{1}{e}\right)\left(1-\frac{1}{\gamma}\right)$.
Consequently, we have 
$R_I \leq R-R\left(1-\frac{1}{e}\right)\left(1-\frac{1}{\gamma}\right)=\theta R$, 
as claimed.

The proof of Theorem~\ref{t:mpec} is complete.

\section{A $\left(k+\frac{1}{2}\right)$-approximation (proof of Theorem~\ref{t:mpec'})}

We say that an edge set $F \subseteq E$ {\em covers} a node set $U \subseteq V$, 
or that $F$ is a {\em $U$-cover}, if $\delta_F(v)\neq \emptyset$ for every $v \in U$. 
Consider the following auxiliary problem:

\vspace{0.2cm}

\noindent
{\sf Restricted Minimum-Power Edge-Cover}  \\
{\em Instance:} \ 
A graph $G=(V,E)$ with edge-costs $\{c(e): e \in E\}$, $U \subseteq V$, and~degree 
\hphantom{\em Instance: } 
bounds $\{\ell_v:v \in U\}$. \\
{\em Objective:} 
Find a power assignment $\{\pi(v):v \in V\}$ that minimizes $\sum_{v \in V} \pi(v)$,
\hphantom{\em Objective:}
such that $\pi(v) \geq \ell_v$ for all $v \in U$, and such that the edge set 
\hphantom{\em Objective:}~$F=\{e=uv\in E: \pi(u),\pi(v)\geq c(e)\}$ covers $U$.

\vspace{0.2cm}

Later, we will prove the following lemma.

\begin{lemma} \label{l:restricted}
{\sf Restricted Minimum-Power Edge-Cover} admits a $3/2$-approximation algorithm.
\end{lemma}

Theorem~\ref{t:mpec'} is deduced from Lemma~\ref{l:restricted} and the following statement. 

\begin{lemma} \label{l:alpha}
If {\sf Restricted Minimum-Power Edge-Cover} admits a $\rho$-ap\-pro\-ximation algorithm,
then {\sf Minimum-Power Edge-Multi-Cover} admits a $(k-1+\rho)$-approxima\-tion algorithm.
\end{lemma}
\begin{proof}
Consider the following algorithm. 
\begin{enumerate}
\item 
Let $\pi(v)$ be the power assignment computed by the 
$\rho$-approximation algorithm for {\sf Restricted Minimum-Power Edge-Cover}  
with $U=\{v \in V:r(v) \geq 1\}$ and bounds $\ell_v = w_v$ for all $v \in U$. 
Let $F=\{e=uv\in E: \pi(u),\pi(v)\geq c(e)\}$.
\item 
For every $v \in V$ let $I_v$ be the edge-set obtained by picking 
the least cost $r_F(v)$ edges in $\delta_{E \setminus F}(v)$ and let 
$I=\cup_{v \in V} I_v$. 
\end{enumerate}
Clearly, $F \cup I$ is a feasible solution to {\sf Minimum-Power Edge-Multi-Cover}.
Let ${\sf opt}$ denote the optimal solution value for {\sf Minimum-Power Edge-Multi-Cover}.
In what follows note that $\pi(V) \leq \rho \cdot {\sf opt}$ and that 
$\sum_{v \in V} w_v \leq {\sf opt}$.

We claim that 
$$p_{I \cup F}(V) \leq \pi(V)+(k-1)\cdot {\sf opt} \ .$$ 
As $\pi(V) \leq \rho \cdot {\sf opt}$, this implies  
$p_{I \cup F}(V) \leq (\rho+k-1) \cdot {\sf opt}$.

For $v \in V$ let $\Gamma_v$ be the set of neighbors of $v$ in the graph $(V,I_v)$. 
The contribution of each edge set $I_v$ to the total power is at most 
$p_{I_v}(\Gamma_v)+p_{I_v}(v)$.
Note that $\pi(v) \geq p_{I_v}(v)$ and $\pi(v) \geq p_F(v)$ for every $v \in V$,
hence $p_{F \cup I_v}(v) \leq \pi(v)$.
This implies 
$$
p_{F \cup I}(V) \leq \sum_{v \in V} (\pi(v) + p_{I_v}(\Gamma_v)) = \pi(V)+\sum_{v \in V} p_{I_v}(\Gamma_v) \ .
$$
Now observe that 
$|\Gamma_v|=|I_v| = r_F(v) \leq k-1$ and that $p_{I_v}(u) \leq w_v$
for every $u \in \Gamma_v$.
Thus 
$$p_{I_v}(\Gamma_v) \leq (k-1) \cdot w_v \ \ \ \ \forall v \in V \ .$$
Finally, using the fact that $\sum_{v \in V} w_v \leq {\sf opt}$, we obtain
$$
p_{F \cup I}(V) \leq \pi(V) + \sum_{v \in V} p_{I_v}(\Gamma_v) \leq 
\pi(V)+(k-1)\sum_{v \in V} w_v \leq \pi(V)+(k-1)\cdot {\sf opt} \ .
$$
This finishes the proof of the lemma. \qed
\end{proof}

In the rest of this section we prove Lemma~\ref{l:restricted}.

We reduce {\sf Restricted Minimum-Power Edge-Cover} to the following problem that admits an exact polynomial time algorithm, c.f. \cite{Sch}.

\vspace{0.2cm}

\noindent
{\sf Minimum-Cost Edge-Cover}: \\
{\em Instance:} \ A multi-graph (possibly with loops) $G=(U,E)$ with edge-costs 
\hphantom{\em Instance: }~$\{c(e):e \in E\}$. \\
{\em Objective:} Find a minimum cost edge-set $F \subseteq E$ that covers $U$. 

\vspace{0.2cm}

Our reduction is not approximation ratio preserving, but incurs a loss of $3/2$ in the appro\-xi\-mation ratio.
That is, given an instance $(G,c,U,\ell)$ of {\sf Restricted Minimum-Power Edge-Cover}, 
we construct in polynomial time an instance $(G',c')$ of {\sf Minimum-Cost Edge-Cover} such that:
\begin{itemize}
\item[(i)] 
For any $U$-cover $I'$ in $G'$ corresponds a feasible 
solution $\pi$ to $(G,c,U,\ell)$ with $\pi(V) \leq c'(I')$.
\item[(ii)] 
${\sf opt}' \leq  3{\sf opt}/2$, where ${\sf opt}$ is an optimal solution value 
to {\sf Restricted Minimum-Power Edge-Cover} and
${\sf opt}'$ is the minimum cost of a $U$-cover in $G'$.
\end{itemize}
Hence if $I'$ is an optimal (min-cost) solution to $(G',c')$, then 
$\pi(V) \leq c'(I') \leq 3{\sf opt}/2$.

Clearly, we may 
set $\ell_v=0$ for all $v \in V \setminus U$.
For $I \subseteq E$ let 
$$D(I)=\sum_{v \in V} \max\{p_I(v)-\ell_v,0\} \ .$$
Here is the construction of the instance $(G',c')$, where $G'=(U,E')$ and $E'$ 
consists of the following three types of edges, where for every edge $e' \in E'$
corresponds a set $I(e') \subseteq E$ of one edge or of two edges.  

\begin{enumerate}
\item 
For every $v \in U$, $E'$ has a loop-edge $e'=vv$ with 
$c'(vv)=\ell_v+D(\{vu\})$ where $vu$ is is an arbitrary chosen minimum cost edge 
in $\delta_E(v)$. \\
Here $I(e')=\{vu\}$.
\item 
For every $uv \in E$ such that $u,v \in U$, $E'$ has an edge $e'=uv$ with 
$c'(uv)=\ell_u + \ell_v +D((\{uv\})$. \\
Here $I(e')=\{uv\}$.
\item 
For every pair of edges $ux,xv \in E$ such that $c(ux) \geq c(xv)$, $E'$ has an edge $e'=uv$ 
with $c'(uv)=\ell_v+\ell_u+D(\{ux,xv\})$. \\
Here $I(e')=\{ux,xv\}$.
\end{enumerate}

\begin{lemma} \label{l:D}
Let $I' \subseteq E'$ be a $U$-cover in $G'$, 
let $I=\cup_{e \in I'} I(e)$, and let $\pi$ 
be a power assignment defined on $V$ by $\pi(v)=\max\{p_I(v),\ell_v\}$.
Then $\pi(V) \leq c'(I')$, $I$ is a $U$-cover in $G$,
and $\pi$ is a feasible solution to $(G,c,U,\ell)$.
\end{lemma}
\begin{proof}
We have that $I$ is a $U$-cover in $G$, by the definition of $I$
and since $I(e')$ covers both endnodes of every $e' \in E'$.  
By the definition of $\pi$, we have that $I \subseteq \{e=uv\in E: \pi(u),\pi(v)\geq c(e)\}$. 
Hence $\pi$ is a feasible solution to $(G,c,U,\ell)$, as claimed.

We prove that $\pi(V) \leq c'(I')$.
For $e'=uv \in E'$ let $\ell(e')=\ell_v$ if $e'$ is a type~1 edge,
and $\ell(e')=\ell_u+\ell_v$ otherwise.
Note that 
$\pi(v)=\max\{p_I(v),\ell(v)\}=\ell_v+\max\{p_I(v)-\ell(v),0\}$, hence
$$\pi(V) = \sum_{v \in U} \ell_v + \sum_{v \in V}\max\{p_I(v)-\ell(v),0\}=
\sum_{v \in U} \ell_v + D(I) \ .$$ 
By the definition of $\ell(e')$ and since $I'$ is a $U$-cover
$\sum_{v \in U} \ell_v \leq \sum_{e' \in I'}\ell(e')$.
Also, $D(I) = D\left(\cup_{e' \in I'}I(e')\right)$, by the definition of $I$.
Thus we have
$$\sum_{v \in U} \ell_v +D(I) \leq 
\sum_{e' \in I'}\ell(e') + D\left(\cup_{e' \in I'}I(e')\right) \ .$$
It is easy to see that 
$$D\left(\cup_{e' \in I'}I(e')\right) \leq \sum_{e' \in I'}D(I(e')) \ .$$
Finally, note that $\ell(e')+D(I(e'))=c'(e')$ for every $e' \in I'$ 
(if $e'$ is a type~1 edge, this follows from our assumption that 
$\ell_v \geq \min\{c(e):e \in \delta_E(v)\}$).
Combining we get
\begin{eqnarray*} 
\pi(V) & =    & \sum_{v \in U} \ell_v +D(I) \leq \\ 
       & \leq & \sum_{e' \in I'}\ell(e') + D\left(\cup_{e' \in I'}I(e')\right) \leq \\
       & \leq & \sum_{e' \in I'}\ell(e')+\sum_{e' \in I'}D(I(e'))= \\
       & =    & \sum_{e' \in I'} \left(\ell(e')+D(I(e'))\right) = \\
       & =    & \sum_{e' \in I'} c'(e') = c'(I') \ .
\end{eqnarray*}
\qed
\end{proof}

\begin{lemma} \label{l:pi}
Let $\{\pi(v):v\in V\}$ be a feasible solution to an instance $(G,c,U,\ell)$
of {\sf Restricted Minimum-Power Edge-Cover}.
Then there exists a $U$-cover $I'$ in $G'$ such that 
$c'(I') \leq 3 \pi(V)/2$.
\end{lemma}
\begin{proof}
Let $I\subseteq \{e=uv\in E: \pi(u),\pi(v)\geq c(e)\}$ be an inclusion minimal $U$-cover.
We may assume that $\pi(v)=\max\{p_I(v),\ell_v\}$ for every $v \in V$.
Since any inclusion minimal $U$-cover is a collection of node disjoint stars,
it is sufficient to prove the statement for the case when $I$ is a star.
Then $I$ has at most one node not in $U$, 
and if there is such a node, then it is the center of the star, if $|I| \geq 2$;
in the case $I$ consists of a single edge $e$, then we define the center of $I$ 
to be the endnode of $e$ in $V \setminus U$ if such exists, 
or an arbitrary endnode of $e$ otherwise.

We define a $U$-cover $I'$ in $G'$, and show that 
\begin{equation} \label{e:pi}
c'(I') \leq \frac{3}{2} \sum_{v \in V}\max\{p_I(v),\ell_v\}=\frac{3}{2}\pi(V) \ .
\end{equation}
Let $v_0$ be the center of $I$ and let $\{v_i:1\leq i\leq d\}$ be the leaves of $I$ 
ordered by descending order of costs $c(v_0v_i) \geq c(v_0v_{i+1})$. 
The $U$-cover $I' \subseteq E'$ is defined as follows.
We cover each pair $v_{2i-1},v_{2i}$, $i=1, \ldots,\lfloor d/2 \rfloor$, by a type~3 edge.
This covers all the nodes except $v_0$, and maybe $v_d$ if $d$ is odd.
We add an additional edge $f$ of type 1 or 2, if there are nodes in $U$
($v_0$ and/or $v_d$) that remain uncovered by the picked type~3 edges. 
Formally, we have the following 4 cases, see Figure~1.
\begin{figure} \label{f:cases}
\centering
\epsfbox{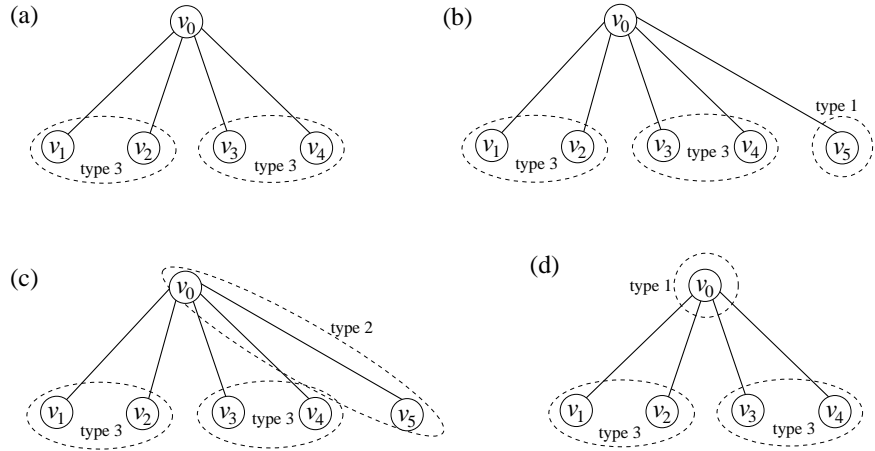}
\caption{Illustration to the definition of the $U$-cover $I'$.}
\end{figure}

\begin{enumerate}
\item
$d$ is even and $v_0 \notin U$, see Figure~1(a). 
Then $U$ is covered by type~3 edges. 
\item
$d$ is odd, and $v_0 \notin U$, see Figure~1(b).  
Then we add a type~1 edge $f$ to cover $v_d$. 
\item
$d$ is odd and $v_0 \in U$, see Figure~1(c). 
Then we add a type~2 edge $f$ to cover $v_0,v_d$.
\item
$d$ is even and $v_0 \in U$, see Figure~1(d). 
Then we add a type~1 edge $f$ to cover $v_0$.
\end{enumerate}

Consider a type~3 edge $v_{2i-1}v_{2i} \in I'$. 
Let $q_i =\max\{c(v_{2i-1}v_0)-\ell_{v_0},0\}$.
Note that $c'(v_{2i-1} v_{2i})\leq \pi(v_{2i-1})+\pi(v_{2i})+q_i$. The key point is that 
$$q_i \leq \frac{1}{2} ( \pi(v_{2i-3})+\pi(v_{2i-2}) ) \ \ \  i=2, \ldots, \lfloor d/2 \rfloor \ .$$
This is since $q_i\leq c(v_0 v_{2i-1})\leq \frac{1}{2}\left(c(v_0 v_{2i-3})+c(v_0 v_{2i-2})\right)$ while $c(v_0 v_j) \leq \pi(v_j) $.
Therefore, 
$$\sum_{i=1}^{d/2} c'(v_{2i-1} v_{2i}) \leq \sum_{i=1}^{d/2}[\pi(v_{2i-1})+\pi(v_{2i})+q_i] \leq \sum_{i=1}^{2 \left\lfloor d/2 \right\rfloor}\pi(v_i)+q_1+\frac{1}{2}\sum_{i=1}^{d-2}\pi(v_i)$$

Now, we prove that (\ref{e:pi}) hold in each one of our four cases.

\begin{enumerate}
\item 
$v_0\notin U$ and $d$ is even. Note that $q_1 \leq c(v_0 v_1) \leq \pi(v_0)$.
Then: $$c'(I')=\sum_{i=1}^{d/2}c'(e_i) \leq \frac{3}{2}\sum_{i=1}^d \pi(v_i) + q_1 \leq \frac{3}{2}\sum_{i=1}^d \pi(v_i) + \pi(v_0) \leq \frac{3}{2}\sum_{i=0}^d \pi(v_i)$$
\item 
$v_0\notin U$ and $d$ is odd. In this case $f=v_d v_d$ is a loop type~1 edge, so 
$c'(f)\leq \pi(v_d)+ \max(c(v_0 v_d)-\ell_{v_0},0)$. 
This implies 
\begin{eqnarray*}
q_1+c'(f) & \leq & c(v_0 v_1) + c(v_0 v_d) + \pi(v_d) \leq \pi(v_0) + \frac{1}{2}[\pi(v_0)+\pi(v_d)]+\pi(v_d) \\
& = & \frac{3}{2}\left(\pi(v_0)+\pi(v_d)\right) \ .
\end{eqnarray*}
Thus
$$c'(I')=\sum_{i=1}^{d/2}c'(e_i) + c'(f) \leq \frac{3}{2}\sum_{i=1}^{d-1} \pi(v_i) + c'(f)+q_1 
\leq \frac{3}{2}\sum_{i=0}^d \pi(v_i)$$
\item 
$v_0\in U$ and $d$ is odd. In this case $f=v_0 v_d$, so 
$c'(f) \leq \max(\ell_{v_0},c(v_0 v_d))+\pi(v_d)$. 
This implies 
$q_1+c'(f)\leq c(v_0 v_1) + c(v_0 v_d) + \pi(v_d) \leq  \frac{3}{2}\left(\pi(v_0)+\pi(v_d)\right)$.
Thus
$$c'(I')=\sum_{i=1}^{d/2}c'(e_i) + c'(f) \leq \frac{3}{2}\sum_{i=1}^{d-1} \pi(v_i) + c'(f)+q_1 
\leq \frac{3}{2}\sum_{i=0}^d \pi(v_i) \ .$$
\item 
$v_0\in U$ and $d$ is even. In this case $f= v_0 v_0$ is a loop type~1 edge, 
so $c'(f)\leq \ell_{v_0}+c(v_0 v_d)\leq \ell_{v_0}+\frac{1}{2}\left(\pi(v_{d-1})+\pi(v_d)\right)$. This implies $q_1+c'(f)\leq c(v_0 v_1)+\frac{1}{2}\left(\pi(v_{d-1})+\pi(v_d)\right)$.
Thus
\begin{eqnarray*}
c'(I') & = & \sum_{i=1}^{d/2}c'(e_i) +c'(f) \leq \sum_{i=1}^d \pi(v_i) +  \frac{1}{2}\sum_{i=1}^{d-2}\pi(v_i) +q_1+c'(f) \\
& \leq & \frac{3}{2}\sum_{i=1}^{d}\pi(v_i) + \pi(v_0) \leq \sum_{i=0}^d \pi(v_i) \ .
\end{eqnarray*}
\end{enumerate}
This concludes the proof of the lemma.
\qed
\end{proof}

As was mentioned, Lemmas \ref{l:D} and \ref{l:pi} imply Lemma~\ref{l:restricted}.
Lemmas \ref{l:restricted} and \ref{l:alpha} imply Theorem~\ref{t:mpec'},
hence the proof of Theorem~\ref{t:mpec'} is now complete.

\section{Conclusions and open problems}

The main result of this paper is a new approximation algorithm for {\sf MPEMC}
with ratio $O(\log k)$. This improves the ratio $O(\log(nk))=O(\log n)$ of \cite{KMNT}.
We also gave a $(k+1/2)$-approximation algorithm, which is better than our 
$O(\log k)$-approximation algorithm for small values of $k$ (roughly $k \leq 6$).

The main open problem is whether the ratio $O(\log k)$ shown in this paper
is tight, or the problem admits a constant ratio approximation algorithm. 

\bibliographystyle{abbrv}
\bibliography{L}

\end{document}